\documentclass[onecolumn,superscriptaddress,floatfix,showpacs]{revtex4-2}
\usepackage{amsmath}
\usepackage{amssymb}
\DeclareMathOperator*{\Max}{Max}
\usepackage{eufrak}
\usepackage[colorlinks=true,linkcolor=blue,citecolor=blue, urlcolor=blue]{hyperref}  
\usepackage{graphicx}
\usepackage[export]{adjustbox}
\usepackage{lipsum}
\usepackage{wrapfig,blindtext}
\usepackage{multirow}
\usepackage{array}
\usepackage{xcolor}
\usepackage{amsmath}
\usepackage{lipsum}
\usepackage{blindtext}
\usepackage{amsthm}
\usepackage{caption}
\usepackage{subcaption}
\usepackage{hyperref}
\usepackage{float}
\usepackage{subfloat}
\newtheorem{theorem}{Theorem}[section]
\newtheorem{corollary}{Corollary}[theorem]

\begin{document}
		\title{Quantum channels and some absolute properties of quantum states}

	\author{Tapaswini Patro}
	\email[]{p20190037@hyderabad.bits-pilani.ac.in}
	
	\affiliation{Department of Mathematics, Birla Institute of Technology and Science-Pilani, Hyderabad Campus, Telangana, India}
	
	\author{Kaushiki Mukherjee}
	\email[]{kaushiki.wbes@gmail.com}
	
	\affiliation{Department of Mathematics, Government Girls' General Degree College, Ekbalpore, Kolkata, India}

	\author{Nirman Ganguly}
	\email[]{nirmanganguly@gmail.com}
	
	\affiliation{Department of Mathematics, Birla Institute of Technology and Science-Pilani, Hyderabad Campus, Telangana, India}

\begin{abstract}
	
Quantum systems are vulnerable against environmental interactions which may result in depletion of potential resources. The decay in the resources may be to the extent that they cannot be retrieved even with a nonlocal unitary action on the whole system. Two important figure of merits in the context of quantum information are the fully entangled fraction (FEF)  and conditional entropy of a composite quantum system. While FEF plays a key role in teleportation, negativity of conditional entropy assumes significance in state merging and dense coding. A state may lose such merits and may move to an absolute regime, where even a global unitary fails to reclaim those merits. In the present work, we probe the action of some quantum channels in two qubits and two qudits and find that some quantum states move to the absolute regime under the action. Since global unitary operations are unable to retrieve them back to the non-absolute regime, we provide a prescription for the retrieval using an entanglement swapping network. We also provide an explicit illustration of our prescription. Furthermore, we extend the notion of absoluteness to conditional Rényi entropies and find the required condition for a state to have absolute conditional Rényi entropy non-negative (ACRENN) property. Exploiting the Bloch-Fano decomposition of density matrices, we characterize such states. We then extend the work to include the marginals of a tripartite system and provide for their characterization with respect to the aforementioned absolute properties.

\end{abstract}
\date{\today}
	
\maketitle

	
\section{Introduction}

The theory of quantum information processing \cite{nie,wildeqit} promises to provide significant advantages over its classical counterpart. In this context, it is important to identify resources unique to quantum theory which can provide for such advantages. Entanglement \cite{ent}, Bell nonlocality \cite{bellnonlocality}, steerability \cite{steerability}, negativity of conditional entropy \cite{mahathi1}, coherence \cite{coherence} are some of them. Investigations on potent resources and methods to harness them have buttressed the study on quantum resources \cite{resource} in general.

However, when quantum systems interact with the environment, such resources may be exposed to depletion. The interaction with environment is modelled in the language of completely positive trace preserving maps, usually termed as quantum channels \cite{nie}. Literature contains a large volume of work on quantum channels, as they are also a medium for sending quantum and classical information (see \cite{survey} and references therein). Regarding the loss of resources, entanglement breaking maps \cite{ent1} have been studied in the context of quantum entanglement. With regard to nonlocality one has the nonlocality breaking maps \cite{ent2,env}.

The extent of decay in the resources might vary across different scenarios. In some situations the quantum state may enter an \textit{absolute regime} during the process. An explanation regarding the term \textit{absolute regime}  warrants attention here. If we take the case of entanglement, then entanglement depends on the choice of the factorization of the underlying composite Hilbert space \cite{Hahn}. In some basis the quantum state may be entangled, while in others it may be separable \cite{volume}. Quite interestingly, there are states which remain separable in any basis, consequently termed as absolutely separable states \cite{ AS}. Precisely, if we start with a bipartite state $ \rho $, and find that $ U \rho U^{\dagger} $ remains separable for any unitary operator $ U $, then $ \rho $ is separable from spectrum \cite{twoqubit,John, AS1} or absolutely separable \cite{AS}. Since, a unitary conjugation preserves the spectrum, the condition for absolute separability is a function of the spectrum of the density matrix, i.e., its eigenvalues. However, necessary and sufficient condition for membership in the absolutely separable class (in terms of eigenvalues) exists only in the qubit-qubit \cite{twoqubit} and qubit-qudit system \cite{John}, the problem in higher dimensions remains open \cite{knill}. Absolutely separating maps were investigated in \cite{separabilitymap}, which can turn a state to being absolutely separable. Thus, in such cases the states enter the \textit{absolute regime} pertaining to entanglement. Studies on absoluteness of quantum correlations were extended beyond entanglement to include absolute positive partial transpose \cite{absoluteppt}, absolute conditional von Neumann entropy nonnegative \cite{spatro, Mahathi}	, absolutely classical spin states \cite{classicalspin}, absolutely $ k- $ incoherent states \cite{k-copy}, Bell-CHSH nonlocality \cite{BCHGU}, absolute non-violation of a three settings steering inequality \cite{steeringinequality}, discord \cite{discord1} and absolute fully entangled fraction \cite{af}. Analytical characterizations of absolutely separable states were provided in \cite{Ganguly2014,extreme}.

Quantum conditional entropies can be negative, unlike the classical entropies. Negative conditional von Neumann entropy provides quantum advantage in state merging \cite{state} and dense coding \cite{dense}. However, if initially a state doesn't have a negative entropy, it may be converted to a state with negative entropy with proper application of an unitary gate (i.e., an unitary transformation). The unitary should be global as local unitary cannot change such characteristics. However, a state belonging to the absolute regime cannot be transformed in this way, an example being the states preserving nonnegative conditional von Neumann entropy under any global unitary action \cite{spatro, Mahathi}. A similar justification goes for fully entangled fraction (FEF). FEF $>1/d$ is a significant benchmark for quantum states to be useful for teleportation \cite{tele,horodecki} and $k-$ copy nonlocality \cite{kcopy}. Interestingly, here too, there are certain states which remain in the absolute regime under the action of global unitaries \cite{af}. 

In quantum information theory, every relevant transition can be modelled in terms of quantum channels and it will not be an overstatement that quantum channels are all-encompassing \cite{wildeqit}. As mentioned before, due to environmental interaction a quantum state may lose some of its important characteristics to the extent that it enters an absolute regime. Our work considers two of the resources mentioned above namely, (a) negative conditional entropy and (b) FEF $>1/d$, and probes instances under which states lose these characteristics \textit{absolutely} due to environmental influence. In the present work, we investigate two scenarios- (i) one involving states that preserve nonnegative conditional von Neumann entropy under any global unitary operation ($ \mathfrak{AC} $) and (ii) the other involving states having an absolute fully entangled fraction ($ \mathfrak{AF} $). We find that quantum states which were previously in the non-absolute regime enter the absolute regime under the action of some quantum channels. The studies are done for both $ 2 \otimes 2 $ and $ d \otimes d $ dimensions. When a state enters the absolute regime, it is not possible to retrieve the concerned resource even through global unitary action. Therefore, we lay a prescription using which one can make the conversion to the non-absolute regime pertaining to $ \mathfrak{AC} $ . The prescription is given using an entanglement swapping network \cite{entswap1,entswap2}. We also provide a demonstration of the prescription.

In the area of quantum information theory, quantum resources have undergone multifaceted analysis. In this context, one important probe has been to identify environmental interactions which may result in depletion of the resource. This may be to the extent where merits are lost in a manner, when even non-local unitary action fails to redeem such merits. That is, as noted before, the state enters an absolute regime. Details of our motivation are now listed below:
 
 Firstly, we identify certain environmental noise which results in the loss of the characteristic. The characteristics considered here are the negativity of conditional von Neumann entropy and the other, FEF being greater than the classical benchmark (i.e., $1/d$). Connection of negative conditional entropy to state merging \cite{state}, dense coding \cite{dense} and FEF $> \frac{1}{d}$ to teleportation \cite{tele}, are well established. In this perspective, our work studies environmental influences which takes an initially potent state to an absolute regime pertaining to the two characteristics mentioned above. In a related perspective one may also note this as the environmental noise that should be avoided during the implementation of the said protocols.
 
  Secondly, since global unitary action cannot create the resource from states in the absolute regime, we provide a prescription using an entanglement swapping network to retrieve the resource. Once a state enters the absolute regime, it becomes \textit{useless} for the concerned quantum information protocol. Therefore, the swapping network protocol (details given in sec.\ref{swapp}) stochastically brings back the state to the non-absolute regime. The non-absolute regime contains both resourceful states and states which can be transformed to being resourceful by the application of a unitary gate (which previously was not possible). We give an explicit demonstration of this fact in sec. \ref{appl}.
 
 Furthermore, we have considered the negativity of conditional Rényi entropy. Negativity of conditional Rényi entropy has been related to Bell nonlocality \cite{Hororenyi} and recently to quantum steering \cite{Steerrenyi}. The discussion above raises the question, whether is it always possible to create a state with negative conditional Rényi entropy from a state having a nonnegative one. The answer is "no", as there are certain states which preserve nonnegativity under any global unitary operation. We also characterize such states here.

 Furthermore, in our work we extend the notion of \textit{absoluteness} to include conditional Rényi entropies. In the asymptotic regime of quantum information, von Neumann entropies remain a pertinent figure of merit. However in its single-shot counterpart, it is the Rényi entropy that assumes significance. We investigate the effect of global unitary action on the conditional Rényi entropy of a quantum system and find the condition under which the non-negativity of conditional Rényi entropy is preserved. Subsequently, the characterization is also done using the Bloch-Fano representation of the quantum state. Using Schur concavity \cite{Renyi1} of the Rényi entropy (for some specified value of a parameter), we also identify the transitions under which the aforementioned absolute property remains invariant. Since, partial trace i.e., discarding a party can also be formulated in the language of quantum channels, we then investigate the reduced subsystems of a tripartite system in the context of absolute conditional Rényi 2-entropy.

		\par The structure of our work is as follows: In sec. \ref{pre}, we provide the the relevant definitions and notations to be used in the work. The transitions from the non-absolute regime to absolute regime are studied in sec.\ref{nec}. Characterization of absolute conditional Rényi entropy non-negative $\mathbf{(ACRENN)}$ property is discussed in \ref{Renyi}. In sec.\ref{tripartite}, we have characterized the membership in $\mathbf{ACRE2NN}$ (states having absolute conditional Rényi 2-entropy) using the Bloch-Fano representation, wherein we also study the marginals of a tripartite system. We discuss the specific entanglement swapping network in sec.\ref{swapp}. An application of our work is provided in sec. \ref{appl}. Finally, we present our conclusions in sec.\ref{con}.


	\section{Preliminaries and Notations}\label{pre}
	In this section, we introduce some preliminary concepts and notations that are required for our study. Here $ \mathsf{H_{AB}} $ denotes the finite-dimensional Hilbert space of the composite system $ \mathsf{AB} $. $ \mathfrak{B}(\mathsf{H}) $ denotes the set of bounded linear operators acting on the Hilbert space $ \mathsf{H} $. The respective notations for different classes of states are given as follows: absolute fully entangled fraction ($\mathfrak{AF}$), absolute conditional von Neumann entropy non-negative ($\mathfrak{AC}$) and absolute conditional Rényi entropy non-negative ($\mathbf{ACRENN}$).
	\subsection{Absolute Fully Entangled Fraction ($\mathfrak{AF}$):} For a density matrix $ \rho_{d \otimes d} \in \mathfrak{B}(\mathsf{H_{AB}}) $, the fully entangled fraction (FEF) is given as,
	\begin{equation}
		F(\rho_{d \otimes d})=\Max\limits_{U_l}\:\langle \psi_{+}|(\mathbb{I}_d\otimes U_l^{\dagger})\: \rho_{d \otimes d} \:(\mathbb{I}_d\otimes U_l)|\psi_{+}\:\rangle
	\end{equation}
	where the maximization is done over all local unitary operators $ U_l $ and $|\psi_{+}\rangle =\frac{1}{\sqrt{d}}\:\sum_{i=0}^{d-1}\:|ii \rangle $ is the maximally entangled state \cite{horodecki}.
	\par $ \mathfrak{F} $ denotes the set of density matrices whose FEF is bounded by $ \frac{1}{d} $, i.e., $ \mathfrak{F} = \lbrace \rho_{d \otimes d} \in \mathfrak{B}(\mathsf{H_{AB}}): F(\rho) \le \frac{1}{d} \rbrace $. If $ U_{nl} $ denotes a non-local unitary operator, then the states having absolute FEF is denoted by $ \mathfrak{AF} = \lbrace \sigma_{af} \in \mathfrak{B}(\mathsf{H_{AB}}): F(U_{nl} \sigma_{af} U_{nl}^\dagger) \le \frac{1}{d}, \forall~ U_{nl} \rbrace $. A state belonging to $\mathfrak{AF} $  has been characterized in terms of its eigenvalues, precisely a state belongs to $\mathfrak{AF} $ if and only if its maximum eigenvalue is $\leq \frac{1}{d}$ \cite{af}.

\subsection{Absolute Conditional Von Neumann Entropy Non Negative States($ \mathfrak{AC} $): }The von Neumann entropy of a quantum state $ \rho_{AB} $ is denoted by $ S(\rho_{AB}) = -\mathrm{Tr}(\rho_{AB} \log_2 \rho_{AB}) $ , with the conditional entropy as $ C(\rho_{AB}) = S(\rho_{AB}) - S(\rho_{B}) $. The set $ \mathfrak{AC} $ represents the states whose conditional entropy remains non-negative even under global unitary action. The characterization of $ \mathfrak{AC} $ in bipartite $2\otimes 2$ dimensions was characterized in \cite{spatro}, which tells that a $\rho_{2\otimes 2}\in \mathfrak{AC}$, iff $S(\rho_{2\otimes 2})\geq 1$. The characterization was extended to $d\otimes d$ dimensions, wherein $\rho_{d\otimes d}\in \mathfrak{AC}$ if and only if  $S(\rho_{d\otimes d})\geq \log_2 d$ \cite{Mahathi}.
\subsection{Majorization}
In many areas of probability and information theory, majorization plays a crucial role, whenever one needs to compare disorder among different systems. Let us suppose that there are two $d$-dimensional real vectors $r$ and $s$. We will say that $r$ is majorized by $s$, which is written as $r\prec s$, if there exist a $d$-dimensional permutation matrix $P_{j}$ and a probability distribution ${p_{j}}$ such that
\begin{align}
r=\sum_{j} p_{j} P_{j} s
\end{align}
\par In other words for $r,s \in \mathbb{R}^{d}$, we define the majorization relation $\prec$ that is $r$ is majorized by $s$ as,
$r\prec s$ iff $\sum_{i=1}^{d}r_{i}=\sum_{i=1}^{d}s_{i}$ and $\forall k\in \left\lbrace 1,2,...d-1\right\rbrace$, $\sum_{i=1}^{k} r_{i}^{\downarrow}\leq \sum_{i=1}^{k} s_{i}^{\downarrow}$ \cite{majorization}, where $ r_{i}^{\downarrow}, s_{i}^{\downarrow} $ denote respectively components of the vectors $ r,s $ arranged into non-increasing order . If $ A $ and $ B $ are Hermitian matrices then we define $ A \prec B $, $ A $ is majorized by $ B $, if $ \lambda(A) \prec \lambda(B) $, where $ \lambda(A) $ is the vector of eigenvalues of $ A $, arranged into non-increasing order \cite{majorization}.
		\subsection{Quantum Channels}
			Quantum channels are completely positive trace-preserving (CPTP) maps transmitting both quantum and classical information. Every quantum channel admits an operator sum representation. Let $\rho$ be a quantum state and $\epsilon$ be a quantum channel. Then the action of the quantum channel $\epsilon$  on the state $\rho$ can be expressed as follows,
			\begin{align}\label{qc}
				\epsilon(\rho)=\sum_{i}K_{i}\rho K_{i}^{\dagger}
			\end{align}
			where $K_{i}'s$ are the Kraus operators satisfying the property $\sum_{i}K_{i}^{\dagger}K_{i}=\mathbb{I}$. In what follows below, we note the actions of the channels important for our study.
		
		 \subsubsection*{ Global depolarizing channel}
The global depolarizing channel is a unique class of quantum channel that transforms a state into a convex combination of itself and the maximally mixed state. i.e., for $\rho_{d \otimes d}$, the action of global depolarizing channel is given as follows
\begin{align}\label{gc}
	\rho_{d \otimes d}^{'}=(1-p)\rho_{d \otimes d}+\frac{p}{d^{2}} \mathbb{I}_{d^2},
\end{align}
where $ p $ is the channel parameter with $0\leq p\leq1$.\\

		\par For two qubit systems, the description of some quantum channels to be used are given in Table \ref{table1} below,
		
		\begin{table}[!h]
\begin{center}
\scalebox{0.8}{%
\begin{tabular}{ || m{4cm} || m{6cm}||  }
  \hline
  \begin{center}
  	\textbf{Quantum channels}
  \end{center} & \begin{center}
  \textbf{Kraus operators}
\end{center} \\
  \hline
\begin{center} $1$.Phase-flip channel \end{center} &  \begin{align*}
 K_0=\sqrt{1-p}\mathbb{I}_2
\end{align*} \begin{align*}
 K_1=\sqrt{1-p}\sigma_{z}
\end{align*}  \\
  \hline
  \begin{center}$2$.Bit-flip channel\end{center} &  \begin{align*}
 K_0=\sqrt{1-p}\mathbb{I}_2
\end{align*} \begin{align*}
 K_1=\sqrt{1-p}\sigma_{x}
\end{align*}  \\
  \hline
  \begin{center}$3$.Phase damping channel\end{center} & \[
  K_{0} =
  \left({\begin{array}{cc}
    1 & 0 \\
    0 & \sqrt{1-p} \\
  \end{array} }\right)
\]   \[
  K_{1} =
  \left( {\begin{array}{cc}
    0 & 0 \\
    0 & \sqrt{p} \\
  \end{array} }\right)
\] \\
  \hline
  \begin{center}$4$.Amplitude damping channel\end{center} & \[
  K_{0} =
  \left({\begin{array}{cc}
    1 & 0 \\
    0 & \sqrt{1-p} \\
  \end{array} }\right)
\]   \[
  K_{1} =
  \left( {\begin{array}{cc}
    0 & \sqrt{p} \\
    0 & 0 \\
  \end{array} }\right)
\] \\
  \hline
 \begin{center} $5$.Depolaring channel\end{center} & \begin{align*}
 K_{0}=\sqrt{1-p} \mathbb{I}_2
\end{align*}  \begin{align*}
K_{1}=\sqrt{\frac{p}{3}}\sigma_{x}
\end{align*} \begin{align*}
K_{2}=\sqrt{\frac{p}{3}}\sigma_{y}
\end{align*}  \begin{align*}
K_{3}=\sqrt{\frac{p}{3}}\sigma_{z}
\end{align*}  \\
  \hline
  \end{tabular}}
 \caption{\label{demo-table} Description of different quantum channels. The action of a quantum channel $\Phi$ on a system $\omega$ is $\Phi(\omega)= \sum\limits_i K_i \omega K_i^\dagger $.}\label{table1}
\end{center}
\end{table}
\section{From the non-absolute to the absolute regime}\label{nec}
In this section, we show that how the action of quantum channels results in the depletion of quantum resources, pertaining to the figure of merits discussed earlier.
\subsection{Pure State in two qubits}
Let us now consider a pure entangled state \cite{schmidt},
\begin{equation}\label{ps}
	|\psi\rangle=\cos\theta|00\rangle+\sin\theta|11\rangle,
\end{equation}
where $\theta\in(0,\frac{\pi}{2})$.

\subsubsection{Global depolarizing channel:}
The action of the global depolarizing channel on the state results in the following transformed state,
\small
\begin{equation}\label{gdc}
	\rho_{2\otimes 2}^{'}(G)=\begin{pmatrix}
		\frac{1}{4}(1+p+2p\cos 2\theta)& 0 & 0 & p\cos\theta\sin\theta\\
		0& \frac{1-p}{4} & 0 &0 \\
		0& 0 &  \frac{1-p}{4}&0 \\
		p\cos\theta\sin\theta& 0 &  0& \frac{1}{4}(1+p-2p\cos 2\theta)
	\end{pmatrix}
\end{equation}
where, $p\in (0,1)$ and $\theta\in (0,\frac{\pi}{2})$.

The eigen values of $\rho_{2\otimes 2}^{'}(G)$ are $\frac{1-p}{4}$, $\frac{1-p}{4}$, $\frac{1-p}{4}$ and  $\frac{1+3p}{4}$, which are independent of $\theta$. The resulting state  $\rho_{2\otimes 2}^{'}(G)$ belong to $\mathfrak{AC}$ and $\mathfrak{AF}$ for the values of $0\leq p\leq 0.747614$ and  $0\leq p\leq 0.3333$ respectively.\\
In case of the amplitude and phase damping channels given below we have used a double interaction meaning that the channel is applied on both the subsystems.
\subsubsection{Amplitude Damping Channel}
The transformed state under the action of the amplitude damping channel is given by,
\begin{widetext}
	\begin{equation}\label{adc}
		\rho_{2\otimes 2}^{'}(Amp)=\begin{pmatrix}
			\cos^2{\theta}+p_{1}p_{2}\sin^2\theta& 0 & 0 & \sqrt{1-p_{1}}\sqrt{1-p_{2}}\cos\theta\sin\theta\\
			0& -p_{1}(-1+p_{2})\sin^2\theta & 0 &0 \\
			0& 0 &  -p_{2}(-1+p_{1})\sin^2\theta&0 \\
			\sqrt{1-p_{1}}\sqrt{1-p_{2}}\cos\theta\sin\theta& 0 &  0& (-1+p_{1})(-1+p_{2})\sin^2\theta
		\end{pmatrix}
	\end{equation}
\end{widetext}
where,  $\theta\in (0,\frac{\pi}{2})$ and $p_{1}$, $p_{2}$ are the channel parameters.

For  $p_{1}=p_{2}=p$ and $\theta=\frac{\pi}{4}$, the eigen values of the state $\rho_{2\otimes 2}^{'}(Amp)$ are $\frac{1}{2}(1+p^2-p+\sqrt{1+2p^2-2p})$,  $\frac{1}{2}(1+p^2-p-\sqrt{1+2p^2-2p})$, $\frac{1}{2}(p-p^2)$ and $\frac{1}{2}(p-p^2)$. The resulting state $\rho_{2\otimes 2}^{'}(Amp)$ belong to $\mathfrak{AC}$ for the values of $0.267284\leq p \leq 0.732716$. It is interesting to note that it never belongs to $\mathfrak{AF}$.

\subsubsection{Phase Damping Channel}
The transformed state under the action of the phase damping channel is given by,
\begin{equation}\label{pdc}
	\rho_{2\otimes 2}^{'}(Ph)=\begin{pmatrix}
		0& 0 & 0 & 0\\
		0& \cos^2\theta & -(-1 + p) \cos\theta \sin\theta &0 \\
		0& -(-1 + p) \cos\theta \sin\theta & \sin^2\theta&0 \\
		0& 0 &  0&0
	\end{pmatrix}
\end{equation}
where,  $\theta\in (0,\frac{\pi}{2})$ and $p$ is the channel parameter.

The eigen values of the state $\rho_{2\otimes 2}^{'}(Ph)$ are $0$,  $0$, $\frac{1}{4}(2-\sqrt{2}\sqrt{2-2p+p^2+2p\cos4\theta-p^2\cos 4\theta})$ and \\$\frac{1}{4}(2+\sqrt{2}\sqrt{2-2p+p^2+2p\cos4\theta-p^2\cos 4\theta})$. The resulting state $\rho_{2\otimes 2}^{'}(Ph)$ does not belong to $\mathfrak{AC}$ as the von Neumann entropy of the resulting state is less than $1$, for the value of $0\leq\theta\leq\frac{\pi}{2}$ and $0\leq p\leq 1$. Hence, this state will never belong to $\mathfrak{AC}$.
\par  The maximum eigen value of $\rho_{2\otimes 2}^{'}(Ph)$ is $\frac{1}{4}(2+\sqrt{2}\sqrt{2-2p+p^2+2p\cos4\theta-p^2\cos 4\theta})$, which is less than $\frac{1}{2}$ for $\theta=\frac{\pi}{4}$ and $p=1$. Hence, it will belong to $\mathfrak{AF}$ for $\theta=\frac{\pi}{4}$ and $p=1$.
\subsection{A mixed state in two qubits}
The following two parameter family of mixed states was considered in \cite{Acin} in the context of hidden nonlocality,
\begin{equation}
	\rho_{2\otimes 2}(\lambda,\theta)=\begin{pmatrix}
		\frac{1-\lambda}{2}& 0 & 0 & 0\\
		0& \lambda\sin^{2}\theta & \frac{\lambda}{2}\sin 2\theta &0 \\
		0& \frac{\lambda}{2}\sin 2\theta &  \lambda\cos^{2}\theta&0 \\
		0& 0 &  0& \frac{1-\lambda}{2}
	\end{pmatrix}
\end{equation}
where, $\lambda\in (0,1)$ and $\theta\in (0,\frac{\pi}{2})$.

\par For the initial choice $\lambda=0.9$ and $\theta=\frac{\pi}{4}$ the state is neither in $\mathfrak{AF}$ nor in $\mathfrak{AC}$. Now, it is observed that for some instances of the channel parameter, the transformed state moves to the absolute regime. We note here that we have again used a double interaction.
	
Table \ref{table2} below depicts the observations, where $R_1$ and $R_2$ respectively refers to the ranges for $ \mathfrak{AC}$  and $ \mathfrak{AF}$.

	\begin{table}[!h]
	\begin{center}
		\begin{tabular}{||c  c  c  c||}
		\hline
		Channels & State Parameter  & $ R_1 $ ~~& $ R_2 $ \\ [0.5ex]
		\hline\hline
		Bit-flip & $\lambda=0.9$ & 0.0890506 $\leq$ p $\leq$ 0.910949 ~~~~~& 0.378732 $\leq$ p $\leq$ 0.621268\\
		\hline
		Phase-flip &  $\lambda=0.9$  &0.0545493 $\leq$ p $\leq$ 0.945451  ~~~~~& 0.333333 $\leq$ p $\leq$ 0.666667  \\
		\hline
		Depolarizing &  $\lambda=0.9$ &-- &  0.271286 $\leq$ p$\leq$1\\
		\hline
		Phase damping &  $\lambda=0.9$  & 0.206295 $\leq$ p $\leq$ 1  ~~~~~&  0.888889 $\leq$ p $\leq$ 1\\
		\hline
		
	\end{tabular}
	\caption{Table depicting the transition to the absolute regime} \label{table2}
	\end{center}
	\end{table}
\subsection{Global depolarizing Channel in two qudits}
 We now consider the isotropic state in two qudits,
 \begin{align}\label{ISO}
 	X_{d\otimes d}(iso)=\beta|\psi^{+}\rangle \langle\psi^{+}|+\frac{1-\beta}{d^2}\mathbb{I}_{d^2},
 \end{align}
 where $|\psi^{+}\rangle =\frac{1}{\sqrt{d}}\sum_{i=0}^{d-1}|ii\rangle $ with $-\frac{1}{d^2-1}\leq\beta\leq 1$. The state $X_{d\otimes d}(iso)$ is separable for $-\frac{1}{d^2-1}\leq \beta\leq \frac{1}{d+1}$ and entangled for $\frac{1}{d+1}< \beta \leq 1$. All entangled isotropic states are useful for teleportation and for all the separable range of parameter $\beta$,  the isotropic state belongs to $ \mathfrak{AF}$ \cite{af}.\\
 After the action of the global depolarizing channel on the isotropic state, the new transformed state is
  \begin{align}\label{iso}
 	X_{d\otimes d}^{'}(iso)=\notag &(1-\lambda)	X_{d\otimes d}(iso)+\frac{\lambda}{d^2}\mathbb{I}_{d^2}\\
 	=&((1-\lambda)\beta)|\psi^{+}\rangle \langle\psi^{+}|+(1-((1-\lambda)\beta))\frac{\mathbb{I}_{d^2}}{d^2}
 \end{align}
 where $\lambda$ is a channel parameter with $0\leq\lambda\leq1$.
 \par  The new state $X_{d\otimes d}^{'}(iso)$ is also an isotropic state with the new parameter $(1-\lambda)\beta $. So, $X_{d\otimes d }^{'}(iso) \in \mathfrak{AF}$, if $(1-\lambda)\beta \leq \frac{1}{1+d}$.

 We now turn our attention to $ \mathfrak{AC}$. Eigen values of  $X_{d\otimes d }^{'}(iso)$ are given as follows:
 $\frac{1+\beta(d^{2}-1)(1-\lambda)}{d^{2}}$ and $\frac{1+\beta(\lambda-1)}{d^{2}}$ with multiplicity ($d^{2}-1$). We know that for $X_{d\otimes d }^{'}(iso)\in\mathfrak{AC}$, its von Neumann entropy should be $\geq \log_2 d$. The von Neumann entropy of the state $X_{d\otimes d }^{'}(iso)$ is given as,
 \begin{align*}
- \frac{(1+\beta(\lambda-1))(d^{2}-1)\log_2[\frac{1+\beta(\lambda-1)}{d^{2}}]+(1-\beta(\lambda-1)(d^{2}-1))\log_2 [\frac{1-\beta(\lambda-1)(d^{2}-1)}{d^{2}}]}{d^{2}}
 \end{align*}
 We now fix the parameter $\beta=0.8$ and then calculate the range of $\lambda$ so that the state $X_{d\otimes d}^{'}(iso)$ belongs to  $\mathfrak{AC}$ . Table \ref{table3} below shows the different ranges of $\lambda$ pertaining to the membership in $\mathfrak{AC}$.
	
	\begin{table}[!h]
	\begin{tabular}{||c  c  c  ||}
		\hline
		Dimension & State Parameter  & Channel parameter  \\ [0.5ex]
		\hline\hline
		$d=2$ & $\beta=0.8$ & $0.0654827 \leq \lambda \leq 1$ \\
		\hline
		$d=3$ &   $\beta=0.8$ &$0.108858  \leq \lambda \leq 1$ \\
		\hline
		$d=4$ &   $\beta=0.8$ &$0.136226  \leq \lambda \leq 1 $ \\
		\hline
		$d=5$ &  $\beta=0.8$  & $0.15533  \leq \lambda \leq1$ \\
		\hline
		\end{tabular}
	\caption{Table depicting the range for $\mathfrak{AC}$ with $\beta=0.8$ } \label{table3}
	
	\end{table}
	In the generic case, it is difficult to obtain an exact value for the von Neumann entropy of quantum states living in arbitrary dimensions. An elegant method to estimate the von Neumann entropy in terms of the trace of different powers of a density matrix was obtained in \cite{AMEVN}. Explicitly, the von Neumann entropy was expressed as an infinite series, given below \cite{AMEVN},
	\begin{align}
	S(\rho)=\sum_{k=1}^{\infty}\frac{g(k)}{k}
	 =(1-R_{2})+\frac{1}{2}(1-2R_{2}+R_{3})+\frac{1}{3}(1-3R_{2}+3R_{3}-R_{4})+\frac{1}{4}(1-4R_{2}+4R_{3}-4R_{4}+R_{5})+...
	\end{align}
	where $g(k)=\sum_{m=0}^{k}\frac{(-1)^{m}k!}{m!(k-m)!}\mathrm{Tr}(\rho^{m+1})$ and $R_{n}=\mathrm{Tr(\rho^{n})}$ with $n\geq 1$. We take the first ten terms of the series to get an estimate of the von Neumann entropy, as given below,
	\begin{equation}\label{aprx}
	 S(\rho)=(1-R_{2})+\frac{1}{2}(1-2R_{2}+R_{3})+...+\frac{1}{10}(1-10R_{2}+10R_{3}-10R_{4}+10R_{5}-10R_{6}+10R_{7}-10R_{8}+10R_{9}-10R_{10}+R_{11})
	\end{equation}

    We now estimate the von Neumann entropy of the transformed state $X_{d\otimes d}^{'}(iso)$ (using Eq.\ref{aprx} ). The table below (Table \ref{table4}), tabulates the ranges of $\lambda$ and  $\beta$ (for $ d=3,4,5,6 $), for which the transformed state will belong to $\mathfrak{AC}$.

	\begin{table}[!h]
	\begin{center}
		\begin{tabular}{||c  c  c   ||}
		\hline
		Dimension & State Parameter  & Channel parameter  \\ [0.5ex]
		\hline\hline
		$d=3$ & $-0.125 \leq \beta \leq 1$ & $0.765349 \leq \lambda \leq 1$ \\
		\hline
		$d=4$ &   $-0.0666667 \leq \beta \leq 1$ &$0.7806 \leq \lambda \leq 1$ \\
		\hline
		$d=5$ &   $-0.0416667 \leq \beta \leq 1$ &$0.831004 \leq \lambda \leq 1$ \\
		\hline
		$d=6$ &   $-0.0285714 \leq \beta \leq 1$ &$0.907309 \leq \lambda \leq 1$ \\
		\hline
		
	\end{tabular}
	\caption{Table depicting the range for $\mathfrak{AC}$ pertaining to the transformed isotropic state.} \label{table4}
	\end{center}
	\end{table}
	
	For $d=2,3,4...10$  with $-\frac{1}{d^2-1}\leq\beta\leq 1$ and $\leq\lambda\leq1$, there exist a range of parameters $\beta$ and $\lambda$ for which   the approximate von Neumann entropy is greater than $\log_{2} d$ (see Fig.\ref{graphiso}).
	
	\begin{figure}[h!]
\includegraphics[width=0.4\linewidth]{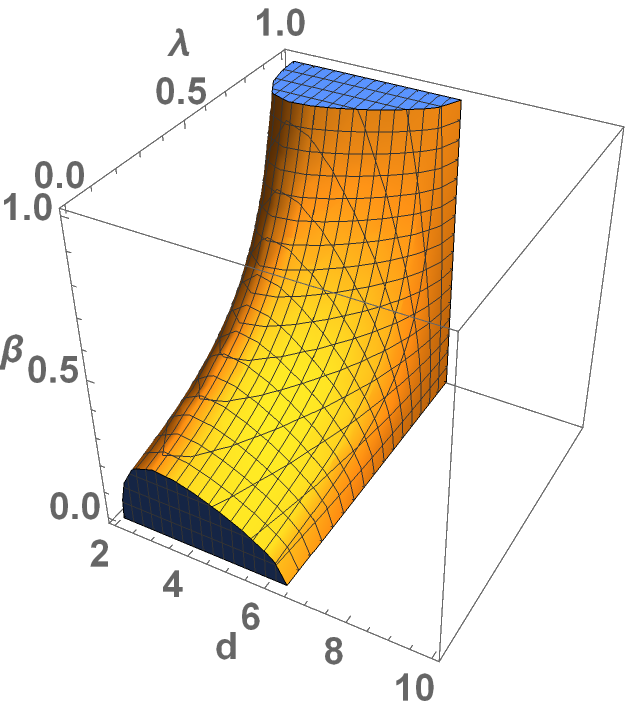}
\caption{Shaded region forms the range of parameters $\beta$ and $\lambda$ and the dimension $d$ for which the von Neumann entropy is $\geq \log_{2} d$  }
\label{graphiso}
\end{figure}

\section{Conditional Rényi Entropy and global unitaries}\label{Renyi}
 In this section, we introduce a class of states for which the conditional Rényi entropy is non-negative under the action of any global unitary operator. The Rényi entropy of order $\alpha$, of a quantum state $ \rho $ is defined as follows \cite{Renyi2},
		\begin{align}\label{renyi}
		S_{\alpha}(\rho)= \frac{1}{1-\alpha}\log_2 (\mathrm{Tr}\rho^{\alpha}),
		\end{align}
		where $\alpha=(0,1)\cup(1,\infty)$. The corresponding conditional Rényi entropy for a state $ \rho_{AB} $is defined as $ S_{\alpha}(A|B)=S_{\alpha}(\rho_{AB})-S_{\alpha}(\rho_{B})$. We note the class of states for which the conditional Rényi entropy is non-negative. We denote the class as \textbf{CRENN}, where \textbf{CRENN} $=\left\lbrace\rho_{cr}:S_{\alpha}(\rho_{cr})_{AB}- S_{\alpha}(\rho_{cr})_{
B}\geq 0\right\rbrace$.
\par Now, the absolute conditional Rényi entropy is defined as the collection of those states whose conditional Rényi entropy remains non-negative under any global unitary operation. We denote the absolute conditional Rényi entropy by $\mathbf{ACRENN}$, where
\begin{align*}
\mathbf{ACRENN}=\left\lbrace\rho_{cr}:S_{\alpha}(U \rho_{cr} U^{\dag})_{AB}- S_{\alpha}(U \rho_{cr} U^{\dag})_{B}\geq 0\right\rbrace
\end{align*}
 For $\alpha=2$, we refer to this class of states as $\mathbf{ACRE2NN}$.
\subsection{$ \mathbf{ACRENN} $}
In what follows below we characterize $ \mathbf{ACRENN} $. However, before doing this we note two observations related to Rényi entropy.

\textbf{Observation-$1$:}\\
\par Let $f$ be a function defined as $f:\mathbb{R}^{d}\rightarrow \mathbb{R}$. $f$ is said to be Schur-concave if $f(x) \leq f(y)$ whenever $x$ majorizes $y$. The Rényi entropy of order $\alpha$ is Schur-concave for all $\alpha$ \cite{Renyi1}. Hence, for  $\rho\succeq \rho^{'}$,
	\begin{align*}
	S_{\alpha}(\rho)\leq S_{\alpha}(\rho^{'})
	\end{align*}
\par \textbf{Observation-$2$:}\\
\par  For a $ d- $ dimensional system, the maximum value of its Rényi entropy is upper bounded by $\log_2 d$ \cite{Renyi1}.\\
The results are now given below,
\begin{theorem}\label{the1}
	A quantum state $\rho\in$ $\mathbf{ACRENN}$  if and only if  $S_{\alpha}(\rho)\geq \log_2 $ $d$
\end{theorem}
\begin{proof}
A quantum state $\rho \in \mathfrak{AC}$ if and only if $S(\rho)\geq \log_{2} d$, which is based on the fact that the maximum von Neumann entropy of a sub-system is $\log_2 d$ \cite{Mahathi}.  We know that the maximum value of Rényi entropy of a sub-system $B$ is $\log_2 d$ \cite{Renyi1}. The proof of this theorem now follows from the Theorem-$4$ of the paper \cite{Mahathi}.
\end{proof}
\begin{theorem}
	If $\rho \succeq \rho^{'}$ and $\rho \in$ $\mathbf{ACRENN}$, then $\rho^{'}\in \mathbf{ACRENN}$
\end{theorem}
\begin{proof}
	Let $\rho \in \mathbf{ACRENN}$. Then $S_{\alpha}(\rho)\geq \log_2 $ $d$. From Schur concavity of the Rényi entropy, it follows that
	\begin{align*}
	S_{\alpha}(\rho)\leq S_{\alpha}(\rho^{'})
	\end{align*}
Hence, the theorem.
\end{proof}
\begin{theorem}\label{Tr}
A quantum state $\rho \in $ $\mathbf{ACRENN}$ iff
\begin{equation*}
\mathrm{Tr}\rho^{\alpha}\begin{cases}
         \geq d^{1-\alpha} \quad &\text{if}~~  \alpha \in  $$ (0 , \; 1)$$ \\
         \leq d^{1-\alpha} \quad &\text{if} ~~ \alpha \in $$(1 , \; \infty)$$  \\
     \end{cases}
\end{equation*}
\end{theorem}
\begin{proof}

\par \textbf{Case-1: $\alpha\in(0 , \; 1)$}

Since $\alpha\in(0 , \; 1)$, $1-\alpha > 0$. We know that $\rho \in $ $\mathbf{ACRENN}$, iff $S_{\alpha}(\rho)\geq \log_2 d$.
\begin{align}
	S_{\alpha}(\rho)\geq &\notag\log_2 d\\
	\iff& \notag \frac{1}{1-\alpha}\log_2 (\mathrm{Tr}\rho^{\alpha})\geq \log_2 d\\
	\iff& \notag \mathrm{Tr}\rho^{\alpha} \geq  d^{1-\alpha}
\end{align}
\par \textbf{Case-2: $\alpha\in(1 , \; \infty)$}

Since $\alpha\in(1 , \; \infty)$, $1-\alpha < 0$. We know that $\rho\in $ $\mathbf{ACRENN}$, iff $S_{\alpha}(\rho)\geq \log_2 d$.
\begin{align}
	S_{\alpha}(\rho)\geq &\notag\log_2 d\\
	\iff& \notag \frac{1}{1-\alpha}\log_2 (\mathrm{Tr}\rho^{\alpha})\geq \log_2 d\\
	\iff& \notag \mathrm{Tr}\rho^{\alpha} \leq  d^{1-\alpha}
\end{align}
\end{proof}

\par Note that for $\alpha=2$, theorem-\ref{Tr} can be written as $\mathrm{Tr}\rho^{2}\leq \frac{1}{d}$, which is nothing but the purity of a state. So, a quantum state $\rho\in$ $\mathbf{ACRE2NN}$ iff its purity $\leq \frac{1}{d}$.

\subsection{Characterization of ACRE2NN in the Bloch-Fano representation}
\par In this section, we characterize absolute conditional Rényi 2 entropy non-negative ($\mathbf{ACRE2NN}$) class in terms of the Bloch-Fano decomposition of a density matrix.
\begin{theorem}\label{the2}
 $ \rho_{d \otimes d} \in \mathfrak{B}(\mathsf{H_{AB}}) $ belongs to $\mathbf{ACRE2NN}$ if and only if
 \begin{align*}
 	\|T\|^{2}\leq  \frac{d^{2}(d-1)-2d(\|\vec{a}\|^{2}+\|\vec{b}\|^{2})}{4}
 \end{align*}
where $T $ is the correlation matrix and $ \vec{a} $, $ \vec{b} $ are the local Bloch vectors.
\end{theorem}
\begin{proof}
  $\rho_{d\otimes d}$  can be expressed as below \cite{Bloch},

\begin{widetext}
 \begin{align}\label{r1}
 	\rho_{d\otimes d}  =\frac{1}{d^{2}}\left(\mathbb{I}_d\otimes\mathbb{I}_d+\sum_{m=1}^{d^2-1} a_{m}(\sigma_{m}\otimes \mathbb{I}_d)+\sum_{n=1}^{d^2-1} b_{n}(\mathbb{I}_d\otimes \sigma_{n})+\sum_{m=1}^{d^2-1} \sum_{n=1}^{d^2-1}t_{mn}(\sigma_{m}\otimes \sigma_{n}) \right)
 \end{align}
\end{widetext}

where $\sigma_{m}$ are the generalized Pauli matrices with the properties $\mathrm{Tr}(\sigma_{m}\sigma_{n})=2\delta_{mn}$ and $\mathrm{Tr}(\sigma_{m})=0$ .

\begin{align}\label{r2}
\mathrm{Tr}\rho_{d\otimes d}^{2} =\frac{1}{d^{4}}\left(d^{2}+2d\sum_{m=1}^{d^2-1} a_{m}^{2}+2d \sum_{n=1}^{d^2-1} b_{n}^{2}+4 \sum_{m=1}^{d^2-1} \sum_{n=1}^{d^2-1}t_{mn}^{2} \right)
\end{align}
where Eq.\ref{r2} follows from the fact that $\mathrm{Tr}(\sigma_{m}\sigma_{n})=2\delta_{mn}$ and $\mathrm{Tr}(\sigma_{m})=0$. The expression can be suitably modified as,
\begin{align}\label{r3}
	\mathrm{Tr}\rho_{d\otimes d}^{2}= \frac{1}{d^{4}}\left(d^{2}+2d \|\vec{a}\|^{2}+2d \|\vec{b}\|^{2}+4\|T\|^{2} \right)
\end{align}
where $ \|T\|^{2}=\mathrm{Tr}(T^{\dag}T)$ and $\|\vec{a}\|$, $\|\vec{b}\|$ stand for the Euclidean norm.

If we denote $2d \|\vec{a}\|^{2}+2d \|\vec{b}\|^{2}+4\|T\|^{2}=R $, then Eq.\ref{r3} can be written as,
\begin{align}
	\mathrm{Tr}\rho_{d\otimes d}^{2}=\frac{d^2+R}{d^{4}}
\end{align}
From Theorem \ref{the1}, we know that a state will belong to $\mathbf{ACRENN}$, if and only if its Rényi entropy exceeds $ \log_2 d$. Therefore,
\begin{align}
	S_{2}(\rho_{d \otimes d})=&\notag-\log_2 (\mathrm{Tr\rho_{d \otimes d}^{2}})\\
	\iff& S_{2}(\rho_{d \otimes d}) =\notag-\log_2 \left(\frac{d^2+R}{d^{4}}\right)\geq \log_2 d\\
	\iff& \notag \log_{2}\left(\frac{d^{4}}{d^2+R}\right)\geq \log_2 d\\
	\iff& \notag R \leq  d^{3}-d^{2}\\
	\iff& \|T\|^{2}\leq  \frac{d^{2}(d-1)-2d(\|\vec{a}\|^{2}+\|\vec{b}\|^{2})}{4}
\end{align}
\end{proof}

  Weyl states are special cases of the general state in $d \otimes d$ dimensions with vanishing Bloch vectors $a_{m}=b_{n}=0$, $\forall$ $m, n$.
  \begin{corollary}
    For a $d \otimes d$ Weyl state, a state $ \in $ $\mathbf{ACRE2NN}$ if and only if
  \begin{align*}
  	\|T\|^{2}\leq  \frac{d^{2}(d-1)}{4}
  \end{align*}
  \end{corollary}

\section{Characterization of Absolute properties for the marginals in a Tripartite System} \label{tripartite}
Partial trace is a pertinent quantum operation specifically in the case when one needs to study the marginals of a composite system. In this section, we study the reduced subsystems of a three-qudit system in the context of their membership in $\mathbf{ACRE2NN}$.
\subsection{Characterization in the Bloch-Fano representation}

 \par Let $\rho_{d\otimes d\otimes d} \in \mathfrak{B}(\mathsf{H_{A}} \otimes \mathsf{H_{B}} \otimes \mathsf{H_{C}}) $ be a tripartite quantum state, which can be expressed as follows \cite{Ming Li}
 \begin{align}
 \rho_{d\otimes d\otimes d}=\nonumber&\frac{1}{d^{3}} \mathbb{I}_d \otimes \mathbb{I}_d\otimes \mathbb{I}_d+\frac{1}{2d^{2}}\left(\sum t_{i}^{1}\sigma_{i}\otimes \mathbb{I}_d \otimes \mathbb{I}_d +\sum t_{j}^{2}\mathbb{I}_d\otimes\sigma_{j}\otimes \mathbb{I}_d+\sum t_{k}^{3}\mathbb{I}_d\otimes \mathbb{I}_d\otimes \sigma_{k}\right)+\\ & \frac{1}{4d}\left(\sum t_{ij}^{12}\sigma_{i}\otimes\sigma_{j}\otimes \mathbb{I}_d+\sum t_{ik}^{13}\sigma_{i}\otimes\mathbb{I}_d\otimes \sigma_{k}+\sum t_{jk}^{23}\mathbb{I}_d\otimes\sigma_{j}\otimes \sigma_{k}\right)+\frac{1}{8}\sum t_{ijk}^{123}\sigma_{i}\otimes\sigma_{j}\otimes\sigma_{k}
 \end{align}
 where $t_{i}^{1}=\mathrm{Tr}(\rho\sigma_{i}\otimes\mathbb{I}_d\otimes \mathbb{I}_d)$, $t_{j}^{2}=\mathrm{Tr}(\rho \mathbb{I}_d\otimes\sigma_{j} \otimes\mathbb{I}_d)$, $t_{k}^{3}=\mathrm{Tr}(\rho \mathbb{I}_d\otimes \mathbb{I}_d \otimes\sigma_{k})$, $t_{ij}^{12}=\mathrm{Tr}(\rho\sigma_{i}\otimes\sigma_{j}\otimes \mathbb{I}_d)$, $t_{ik}^{13}=\mathrm{Tr}(\rho\sigma_{i}\otimes \mathbb{I}_d\otimes\sigma_{k})$, $t_{jk}^{23}=\mathrm{Tr}(\rho \mathbb{I}_d\otimes \sigma_{j}\otimes\sigma_{k} )$, $t_{ijk}^{123}=\mathrm{Tr}(\rho\sigma_{i}\otimes\sigma_{j}\otimes\sigma_{k})$ with the property $\mathrm{Tr}(\sigma_{i})=0$, $\mathrm{Tr}(\sigma_{i}\sigma_{j})=2\delta_{ij}$. Here, $T_{x}$, $T_{xy}$, $T_{123}$ are the vectors with entries $t_{i}^{x}$, $t_{ij}^{xy}$ and $t_{ijk}^{123}$ respectively and satisfying the condition $1\leq x<y<z\leq 3$.
 \par Removing the first qubit, the reduced state is given by:
 \begin{align}\label{rhobc}
  \rho_{BC}=\frac{1}{d^2}\mathbb{I}_d\otimes \mathbb{I}_d+\frac{1}{2d}\left(\sum t_{j}^{2}(\sigma_{j}\otimes \mathbb{I}_d)+\sum t_{k}^{3}(\mathbb{I}_d \otimes \sigma_{k})\right)+\frac{1}{4}\sum t_{jk}^{23}(\sigma_{j}\otimes\sigma_{k})
 \end{align}
 Now, \begin{align}\label{1}
 \mathrm{Tr}\rho_{BC}^{2}&=\nonumber \frac{1}{d^2}+\frac{1}{2d}\left(\sum \sum t_{j}^{2} t_{j}^{2}+\sum \sum t_{k}^{3} t_{k}^{3}\right)+\frac{1}{4}\sum\sum t_{jk}^{23}t_{jk}^{23}\\
 &=\frac{1}{d^2}+\frac{1}{2d}\left(\|T_{2}\|^{2}+\|T_{3}\|^{2}\right)+\frac{1}{4}\|T_{23}\|^{2}
 \end{align}
Where $\sum t_{j}^{2}t_{j}^{2}=\|T_{2}\|^{2}$,  $\sum t_{k}^{3}t_{k}^{3}=\|T_{3}\|^{2}$ , $\sum t_{jk}^{23}t_{jk}^{23}=\|T_{23}\|^{2}$ and Eq.$\ref{1}$ follows from the fact that $\mathrm{Tr}(\sigma_{i})=0$, $\mathrm{Tr}(\sigma_{i}\sigma_{j})=2\delta_{ij}$.\\
\par We know that a state will belong to $\mathbf{ACRE2NN}$ iff $S_{2}(\rho_{BC})\geq \log_{2} d$.\\
\par Now, let us suppose that $m_{1}=\|T_{2}\|^{2}+\|T_{3}\|^{2}$. Then Eq.\ref{1} can we written as
\begin{align}
\mathrm{Tr}\rho_{BC}^{2}=\frac{1}{d^{2}}+\frac{1}{2d}m_{1}+\frac{1}{4}\|T_{23}\|^{2}
\end{align}
\begin{align}
	S_{2}(\rho_{BC})=&\notag-\log_{2}(\mathrm{Tr\rho_{BC}^{2}})\\
	\iff& S_{2}(\rho_{AB}) =\notag-\log_{2}\left(\frac{1}{d^2}+\frac{1}{2d}m_{1}+\frac{1}{4}\|T_{23}\|^{2}\right)\geq \log_{2} d\\
	\iff& \notag -\log_{2}\left(\frac{4+2dm_{1}+d^{2}\|T_{23}\|^{2}}{4d^2}\right)\geq \log_{2} d\\
	\iff& \notag \log_{2}\left(\frac{4d^{2}}{4+2dm_{1}+d^{2}\|T_{23}\|^{2}}\right)\geq \log_{2} d\\
	\iff& \notag 2m_{1}+d\|T_{23}\|^{2}\leq \frac{4}{d}(d-1)\\
	\iff& \|T_{23}\|^{2}\leq \frac{4(d-1)-2m_{1}d}{d^{2}}
\end{align}
\par Similarly, removing the second qubit, the reduced state is given as follows:
 \begin{align}\label{rhoac}
  \rho_{AC}=\frac{1}{d^2}\mathbb{I}_d\otimes \mathbb{I}_d+\frac{1}{2d}\left(\sum t_{i}^{1}(\sigma_{i}\otimes \mathbb{I}_d)+\sum t_{k}^{3}(\mathbb{I}_d \otimes \sigma_{k})\right)+\frac{1}{4}\sum t_{ik}^{13}(\sigma_{i}\otimes\sigma_{k})
 \end{align}
 \par And the condition for ACRE2NN is given by
 \begin{align}
 \|T_{13}\|^{2}\leq \frac{4(d-1)-2m_{2}d}{d^{2}}
 \end{align}
where $\|T_{13}\|^{2}=\sum t_{ik}^{13}t_{ik}^{13} $ and$~ m_{2}=\|T_{1}\|^{2}+\|T_{3}\|^{2}$.
 \par Similarly, removing the third qubit, the reduced state is given as follows:
 \begin{align}\label{rhoab}
  \rho_{AB}=\frac{1}{d^2}\mathbb{I}_d\otimes \mathbb{I}_d+\frac{1}{2d}\left(\sum t_{i}^{1}(\sigma_{i}\otimes \mathbb{I}_d)+\sum t_{j}^{2}(\mathbb{I}_d \otimes \sigma_{k})\right)+\frac{1}{4}\sum t_{ij}^{12}(\sigma_{i}\otimes\sigma_{j})
 \end{align}
 \par And the condition for $\mathbf{ACRE2NN}$ is give by
 \begin{align}
 \|T_{12}\|^{2}\leq \frac{4(d-1)-2m_{3}d}{d^{2}}
 \end{align}
where $\|T_{12}\|^{2}=\sum t_{ij}^{12}t_{ij}^{12} $ and$~~ m_{3}=\|T_{1}\|^{2}+\|T_{2}\|^{2}$.

\subsection{Pure states in three qubits}

In this section, we study the reduced subsystems of a three-qubit pure state in the context of their membership in $\mathbf{ACRE2NN}$.

In \cite{Acin}, a pure tripartite state is given as,
\begin{equation}\label{pure}
	|\psi_{3}\rangle=x_0 |000\rangle+x_1 e^{\imath \theta}|100\rangle+x_2 |101\rangle+x_3 |110\rangle+x_4|111\rangle
\end{equation}
where $x_i$$\geq$$0,$ $0$$\leq$$\theta$$\leq$$\pi$ and $\sum_{i=0}^4 x_i^2$$=$$1$.\\
Removing the first qubit, the reduced state is given by:
\begin{equation}\label{red1}
	\left(\begin{array}{cccc}
		x_0^2+x_1^2&e^{\imath \theta}x_1x_2&e^{\imath \theta}x_1x_3&e^{\imath \theta}x_1x_4\\
		e^{-\imath \theta}x_1x_2&x_2^2&x_2x_3&x_2x_4\\
		e^{-\imath \theta}x_1x_3&x_2x_3&x_3^2&x_3x_4\\
		e^{-\imath \theta}x_1x_4&x_2x_4&x_3x_4&x_4^2\\
	\end{array} \right)
\end{equation}
Eigenvalues of the reduced state (Eq.(\ref{red1})) are,  $0,0,\frac{1}{2}(1\pm\sqrt{S_1}),$  where $S_1$ is given by $1+4(x_2^2+x_3^2-x_4^2+x_1^2x_4^2+x_3^2(-1+x_1^2+2x_4^2)
+x_2^2(-1+x_1^2+2x_3^2+2x_4^2))$.

Removing the second qubit, the reduced state is given by:
\begin{equation}\label{red2}
	\left(\begin{array}{cccc}
		x_0^2&0&e^{-\imath \theta}x_0 x_1&x_0x_2\\
		0&0&0&0\\
		e^{\imath \theta}x_0x_1&0&x_1^2+x_3^2&e^{\imath \theta}x_1x_2+x_3x_4\\
		x_0x_2&0&e^{-\imath \theta}x_1x_2+x_3x_4&x_2^2+x_4^2\\
	\end{array} \right)
\end{equation}
The eigenvalues are $0,0,\frac{1}{2}(1\pm\sqrt{S_2}),$  where $S_2$ is given by
$ 1-4(x_0^2x_3^2+x_2^2x_3^2-2\cos(\theta)x_1x_2x_3x_4+x_0^2x_4^2+x_1^2x_4^2)$.

Removing the third qubit, the reduced state is given by:
\begin{equation}\label{red3}
	\left(\begin{array}{cccc}
		x_0^2&0&e^{-\imath \theta}x_0x_1&x_0x_3\\
		0&0&0&0\\
		e^{\imath \theta}x_0x_1&0&x_1^2+x_2^2&e^{\imath \theta}x_1x_3+x_2x_4\\
		x_0x_3&0&e^{-\imath \theta}x_1x_3+x_2x_4&x_3^2+x_4^2\\
	\end{array} \right)
\end{equation}
The eigenvalues are $0,0,\frac{1}{2}(1\pm\sqrt{S_3}),$  where $S_3$ is given by
$  1-4(x_0^2x_2^2+x_2^2x_3^2-2\cos(\theta)x_1x_2x_3x_4+x_0^2x_4^2+x_1^2x_4^2)$

The reduced state of the pure state (Eq.\ref{pure}) with respect to first, second and third qubit are given by Eq.\ref{red1}, Eq.\ref{red2} and Eq.\ref{red3} respectively. We know from the theorem \ref{Tr}, for membership in $\mathbf{ACRE2NN}$, $\mathrm{Tr}\rho^{2}\leq \frac{1}{2}$.
\par For the reduced state when first qubit is traced out (\ref{red1}), $\mathrm{Tr}\rho^{2}= x_{0}^{4}+2x_{0}^{2}x_{1}^{2}+(1-x_{0}^{2})^{2}$. The reduced state $\in \mathbf{ACRE2NN}$, if
\begin{align}
x_{0}^{4}+2x_{0}^{2}x_{1}^{2}+(1-x_{0}^{2})^{2}\leq \frac{1}{2}
\end{align}
\par  For the reduced state when second qubit is traced out (\ref{red2}), $\mathrm{Tr}\rho^{2}= x_{0}^{4}+x_{1}^{4}+x_{2}^{4}+x_{3}^{4}+x_{4}^{4}+2x_{0}^{2}(x_{1}^{2}+x_{2}^{2})+2x_{1}^{2}(x_{2}^{2}+x_{3}^{2})+2x_{4}^{2}(x_{3}^{2}+x_{2}^{2})+4x_{1}x_{2}x_{3}x_{4}\cos(\theta)$. The reduced state $\in \mathbf{ACRE2NN}$, if
\begin{align}
x_{0}^{4}+x_{1}^{4}+x_{2}^{4}+x_{3}^{4}+x_{4}^{4}+2x_{0}^{2}(x_{1}^{2}+x_{2}^{2})+2x_{1}^{2}(x_{2}^{2}+x_{3}^{2})+2x_{4}^{2}(x_{3}^{2}+x_{2}^{2})+4x_{1}x_{2}x_{3}x_{4}\cos(\theta)\leq \frac{1}{2}
\end{align}
\par  For the reduced state when third qubit is traced out (\ref{red3}), $\mathrm{Tr}\rho^{2}= x_{0}^{4}+x_{1}^{4}+x_{2}^{4}+x_{3}^{4}+x_{4}^{4}+2x_{0}^{2}(x_{1}^{2}+x_{3}^{2})+2x_{1}^{2}(x_{2}^{2}+x_{3}^{2})+2x_{4}^{2}(x_{3}^{2}+x_{2}^{2})+4x_{1}x_{2}x_{3}x_{4}\cos(\theta)$. The reduced state $\in \mathbf{ACRE2NN}$, if
\begin{align}
x_{0}^{4}+x_{1}^{4}+x_{2}^{4}+x_{3}^{4}+x_{4}^{4}+2x_{0}^{2}(x_{1}^{2}+x_{3}^{2})+2x_{1}^{2}(x_{2}^{2}+x_{3}^{2})+2x_{4}^{2}(x_{3}^{2}+x_{2}^{2})+4x_{1}x_{2}x_{3}x_{4}\cos(\theta)\leq \frac{1}{2}
\end{align}

\subsection{An example for a mixed state in three qubits}
	
	 Consider a mixture of GHZ \cite{GHZ} and W states \cite{W} in three qubits,
\begin{eqnarray}\label{red4}
	\varrho_{2\otimes 2\otimes 2}&=&p |\phi_1\rangle\langle \phi_1| +(1-p) |\phi_2\rangle\langle \phi_2|,\,\textmd{where},\nonumber\\
	|\phi_1\rangle &=&\frac{1}{\sqrt{2}}(|000\rangle+|111\rangle)\nonumber\\
	|\phi_2\rangle &=&\frac{1}{\sqrt{3}}(|001\rangle+|010\rangle+|100\rangle)
\end{eqnarray}
Reduced state by eliminating any one of the three qubits is given by:
\begin{equation}\label{red5}
	\left(\begin{array}{cccc}
		\frac{2+p}{6}&0&0&0\\
		0&\frac{1-p}{3}&\frac{1-p}{3}&0\\
		0&\frac{1-p}{3}&\frac{1-p}{3}&0\\
		0&0&0&\frac{p}{2}\\
	\end{array} \right)
\end{equation}
Eigen values of the reduced matrix are,  $0,\frac{2}{3}(1-p),\frac{p}{2},\frac{2+p}{6}.$ Therefore, the reduced state $\in\mathbf{ACRE2NN}$ for $p\in \left[0.07,1\right]$.

\section{Swapping Network Protocol}\label{swapp}

In this section, we lay a prescription wherein a state which has moved to the absolute regime can be transformed into a state in the non-absolute regime.

The protocol used here for retrieving resources rely  primarily on an entanglement swapping network that entangles two never-interacting pairs of particles using appropriate measurements \cite{entswap1,entswap2}. It consists of two phases: Preparatory Phase and Measurement Phase.

\textbf{ Preparatory Phase:}
Let us suppose that there are three parties Alice , Bob and Charlie and two sources $S_1$ and $S_2$ (see Fig.\ref{fig:swap}). Let each of the two sources $S_1$ and $S_2$ generate an entangled state: $\rho_{AB}$ and $\rho_{BC}$ respectively. Entangled state $\rho_{AB}$ is shared between Alice and Bob whereas  $\rho_{BC}$ is shared between Bob and Charlie. So, Bob receives two particles from the two sources $S_1$ and $S_2$. Then Bob performs full Bell-basis measurement on the joint state of the particles received from $S_1$ and $S_2.$ After performing the measurement, he broadcasts the  results to Alice and Charlie. Let he send bit string $(00),(01),(10),(11)$ when the joint state of the two particles gets projected along $|\phi^{+}\rangle\langle \phi^{+}|,|\phi^{-}\rangle\langle \phi^{-}|,|\psi^{+}\rangle\langle \psi^{+}|$ and $|\psi^{-}\rangle\langle \psi^{-}|$ respectively. Based on Bob's result $(i,j)$, Alice and Charlie now share one of the four conditional states $\rho^{(ij)}.$

\textbf{ Measurement Phase:} In this stage, Alice and Charlie perform suitable global unitary operations on 
 $\rho^{(ij)}$ and then local projective measurements on their respective subsystems. Resulting correlations are now tested to see whether at least one of these conditional states ($\rho^{(ij)}$) $ \in $ $\mathfrak{AC}$ or not. \\
A pure entangled state is first passed through a noisy channel. Let the resulting state be a member of $\mathfrak{AC}.$  Let two copies of such state be then used in the protocol.  In case at least one of the four conditional states $\rho^{(ij)}$ is not a member of $\mathfrak{AC},$ then the protocol is said to have succeeded probabilistically in retrieving the nonclassical feature of the pure entangled state.
The procedure of our protocol is described in Figure \ref{fig:swap}.
\begin{figure}[h]
	\centering
	\includegraphics[width=0.5\linewidth]{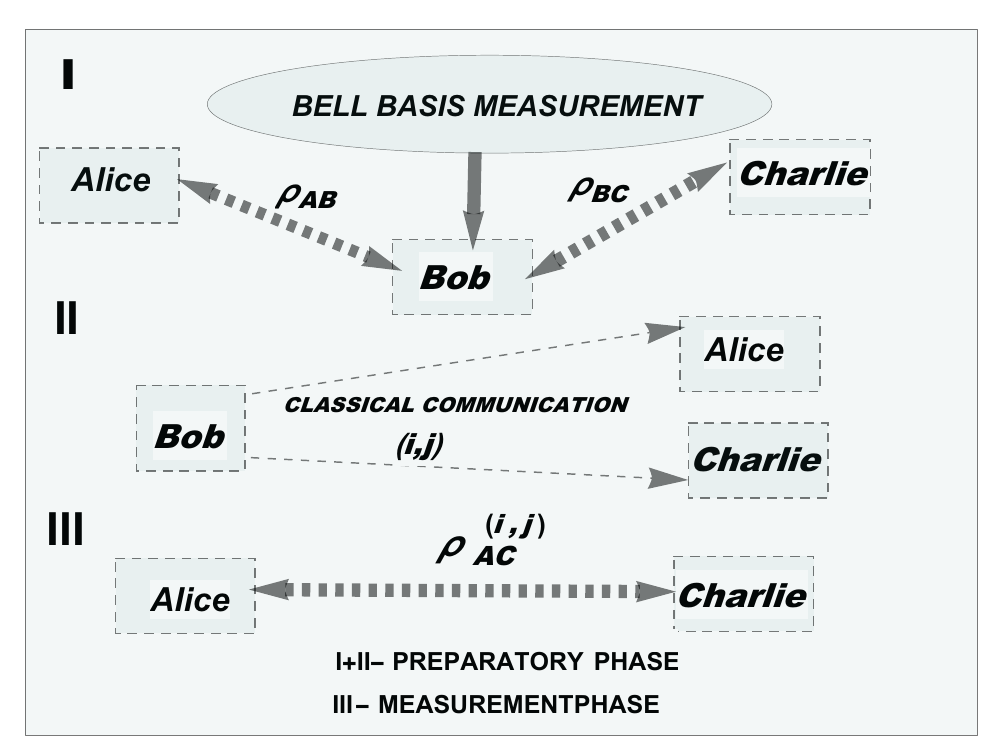}
	\caption{The Swapping Network}
	\label{fig:swap}
\end{figure}
\subsection{Werner state}
The state obtained after passing the two qubit pure entangled state (Eq.\ref{ps}) through the global depolarizing channel is given by Eq.\ref{gdc}. $2$ copies of $\rho_{2\otimes 2}^{'}(G)$ states with parameters ($\theta_{1},p_{1}$) and ($\theta_{2},p_{2}$) are used in the protocol. At the end of preparation phase, four possible conditional states may be generated.  Entropy for each of $ \rho^{(00)}$ and $ \rho^{(01)}$ is same. For $p_{2}=0.705882$, there exists a range of parameters, $0\leq p_1 \leq 1$ and $\theta_{1}, \theta_{2}\in(0,\frac{\pi}{2}),$ for which entropy of each of the two copies used in swapping network have entropy greater than 1, whereas entropy of the of $ \rho^{(00)}$ or $ \rho^{(01)}$ is $\leq 1$ (see Fig.\ref{D12}).

Entropy for each of $ \rho^{(10)}$ and $ \rho^{(11)}$ is same. For $p_2=0.705882$, there exists a range of parameters, $0\leq p_1 \leq 1$ and $\theta_{1}, \theta_{2}\in(0,\frac{\pi}{2}),$ for which entropy of each of the two copies used in swapping network have entropy greater than 1, whereas entropy of $ \rho^{(10)}$ or $ \rho^{(11)}$  is $\leq 1$ (see Fig.\ref{D34}).
\begin{figure}[!ht]
	\centering
	\begin{subfigure}[b]{0.3\textwidth}
		\centering
		\includegraphics[width=\textwidth]{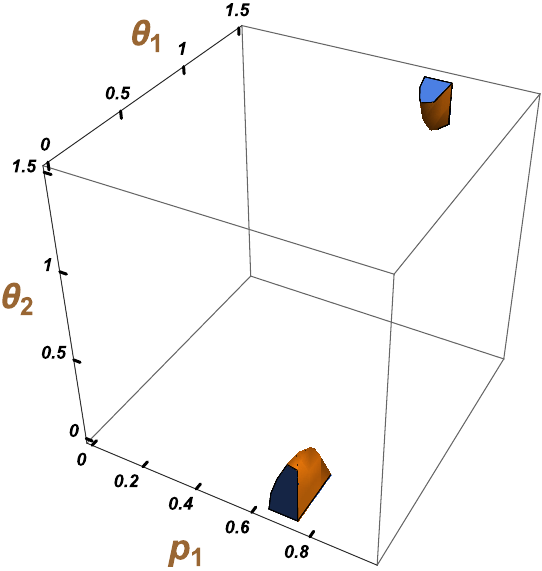}
		\caption{Shaded region gives the range of damping parameter $p_1$ and state parameters $\theta_1,\theta_2$ for which retrieval of nonclassical feature is possible probabilistically.  }
		\label{D12}
	\end{subfigure}
	\hfill
	\begin{subfigure}[b]{0.3\textwidth}
		\centering
		\includegraphics[width=\textwidth]{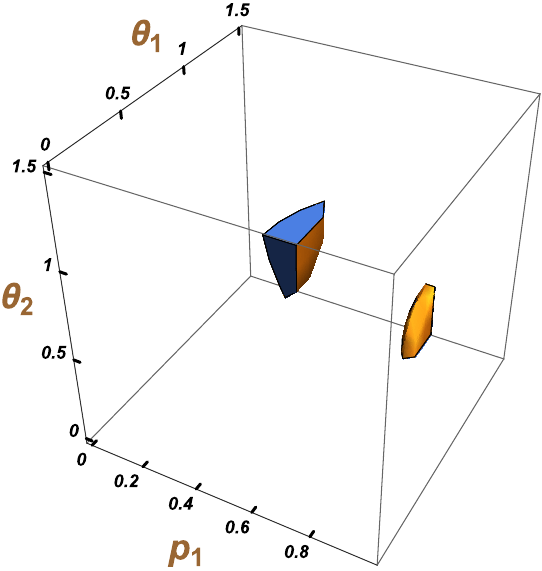}
		\caption{ Shaded region gives the range of damping parameters for which entropy of $ \rho^{(00)}$ and $ \rho^{(01)}$ is less than $1$ whereas that of the initial states is greater than $1$.}
		\label{D34}
	\end{subfigure}
\caption{Absolute to non-absolute regime for the globally depolarized state (Eq.\ref{gdc})}
\end{figure}
\newpage
\subsection{Amplitude damped state}
We have applied the amplitude damping channel over each qubit on the pure state Eq.\ref{ps}. The new state is $\rho_{2\otimes 2}^{'}(Amp)$  given by Eq.\ref{adc}. We have fixed $\theta_{1}=\frac{\pi}{4}$. Here, we have used $2$ copies of $\rho_{Amp}^{'}$ state with the parameters ($p_1,p_2$) and ($p_3,p_4$). So, four conditional states are possible. Entropy for each of $ \rho^{(00)}$ and $ \rho^{(01)}$ is same. There exists a range of damping parameters for which entropy of each of the two copies used in the protocol is greater than 1, whereas entropy of $ \rho^{(00)}$ and $ \rho^{(01)}$ is $\leq 1$ (see Fig.\ref{A12 }). Here, $p_4=0.714286$ and other parameters vary between $0$ and $1$.

Entropy for each of $ \rho^{(10)}$ and $ \rho^{(11)}$ is same. There exists a range of damping parameters for which entropy of each of the two copies used in the protocol is greater than 1, whereas
 entropy of $ \rho^{(10)}$ and $ \rho^{(11)}$ is $\leq 1$ (see Fig.\ref{A34}). Here, $p_4=0.714286$ and other parameters vary between $0$ and $1$.

\begin{figure}[!ht]
	\centering
	\begin{subfigure}[b]{0.3\textwidth}
		\centering
		\includegraphics[width=\textwidth]{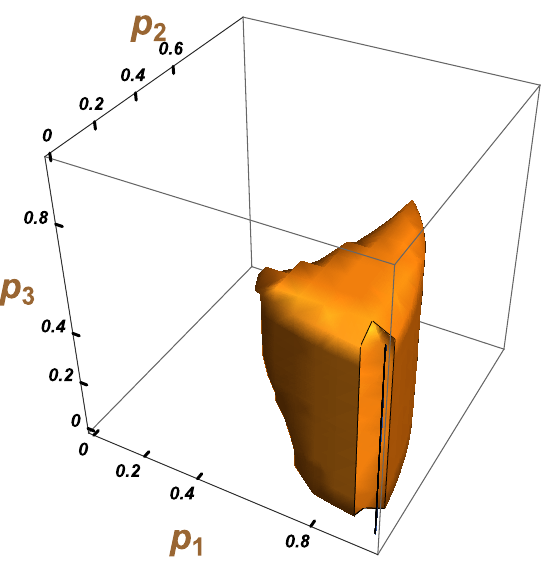}
		\caption{Shaded region gives the range of damping parameters for which entropy of the $ \rho^{(00)}$ and $ \rho^{(01)}$ is less than $1$ but that of each of the two initial states is greater than $1$.}
		\label{A12 }
	\end{subfigure}
	\hfill
	\begin{subfigure}[b]{0.3\textwidth}
		\centering
		\includegraphics[width=\textwidth]{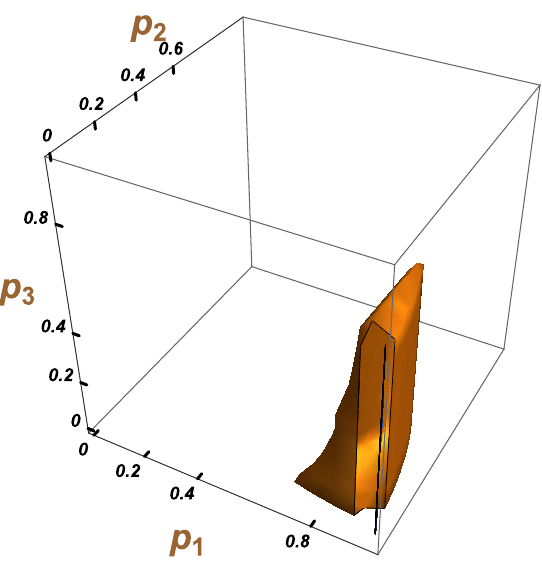}
		\caption{Shaded region gives the range of damping parameters for which entropy of $ \rho^{(10)}$ and $ \rho^{(11)}$ is less than $1$ whereas that of each of the two initial states is greater than $1$.}
		\label{A34}
	\end{subfigure}
	\caption{Absolute to non-absolute regime for the amplitude damped state (Eq.\ref{adc})}
\end{figure}
\subsection{Phase damped state}

We have applied the phase damping channel over each qubit on the pure state Eq.\ref{ps}. The new state is $\rho_{2\otimes 2}^{'}(Ph)$ given by Eq.$\ref{pdc}$.  Here, the entropy of the resultant state will always be less than $1$, so it will never belong to $\mathfrak{AC}$ class. Therefore, there is no need to apply the swapping protocol.

\section{Application} \label{appl}

In this section, we give explicit instances where resources can be retrieved using the swapping network and consequent global unitary action. 

Let us consider the case of negative conditional entropy. As noted before in sec.\ref{swapp}, the swapping protocol can bring the state back into the non-absolute regime. However, the non-absolute regime contains states with negative conditional entropy together with states which possess positive entropy. For the states with positive entropy we need the application of an additional global unitary operator which can give us a state with conditional negative entropy. 

For an illustration, we consider the case of amplitude damping. We have used two non-identical copies of the amplitude damped state $\rho^{'}_{Amp}$ as initial states for the swapping network protocol. The two states are given as below:
\begin{align}
\rho_{1}^{int}= 0.54|00\rangle\langle 00|+0.212132|11\rangle\langle 00|+0.36|01\rangle\langle 01|+0.01|10\rangle\langle 10|+0.212132|00\rangle\langle 11|+0.09|11\rangle\langle 11|
\end{align}
\begin{align}
\rho_{2}^{int}= 0.571429|00\rangle\langle 00|+0.239046|11\rangle\langle 00|+0.0285714|01\rangle\langle 01|+0.285714|10\rangle\langle 10|+0.239046|00\rangle\langle 11|+0.114286|11\rangle\langle 11|
\end{align}
$\rho_{1}^{int}$ is obtained when two qubits of the Bell state $|\phi^{+}\rangle$ are passed through amplitude damping channels with $p_1$$=$$0.8,$ and $p_2$$=$$0.1.$ $\rho_{2}^{int}$ is obtained when two qubits of $|\phi^{+}\rangle$ are passed through amplitude damping channels with $p_3$$=$$0.2,$ $p_4$$=$$0.714286.$ \\
The von-Neumann entropy of $\rho_{1}^{int}$ and $\rho_{2}^{int}$ are $1.06433$ and $1.12407$ respectively, which is $\geq 1$.

The expression of the conditional state obtained on Bob's qubits being projected in Bell state ($|\psi^{+}\rangle$) is
\begin{align}
\rho_{1}^{final}= 0.734694|00\rangle\langle 00|+0.146939|01\rangle\langle 01|+0.103488|10\rangle\langle 01|+0.103488|01\rangle\langle 10|+0.110787|10\rangle\langle 10|+0.00758017|11\rangle\langle 11|
\end{align}
Now, entropy of $\rho_{1}^{final}$ is $0.99883$. Hence, it does not belong to $\mathfrak{AC}$. But, the conditional entropy of $\rho_{1}^{final}$ is $0.377795$, which is not negative. Let us now consider a global unitary operator,

\begin{equation}\label{unitary}
U=\left(\begin{array}{cccc}
		\frac{1}{\sqrt{2}}&0&0&-\frac{1}{\sqrt{2}}\\
		0&1&0&0\\
		0&0&1&0\\
		\frac{1}{\sqrt{2}}&0&0&\frac{1}{\sqrt{2}}\\
	\end{array} \right)
\end{equation}
The application of the global unitary operator $U$ results a modified state $\rho_{new}^{1}=U\rho_{1}^{final}  U^{\dagger}$, given by,
\begin{align}
\rho_{new}^{1}=0.372025|00\rangle\langle 00|+0.364445|00\rangle\langle 11|+0.146939|01\rangle\langle 01|+0.103488|01\rangle\langle 10| \nonumber \\ + 0.103488|10\rangle\langle 01|+0.110787|10\rangle\langle 10|+0.364445|11\rangle\langle 00|+0.372025|11\rangle\langle 11|,
\end{align}
 and the conditional von-Neumann entropy is -0.000227186.
 Similarly, the expression of the conditional state obtained on Bob's qubits being projected in Bell state ($|\phi^{+}\rangle$) is
\begin{align}
\rho_{2}^{final}= 0.806723|00\rangle\langle 00|+0.0994299|11\rangle\langle 00|+0.110924|01\rangle\langle 01|+0.0616247|10\rangle\langle 10|+0.0994299|00\rangle\langle 11|+0.0207283|11\rangle\langle 11|
\end{align}
Entropy of $\rho_{2}^{final}$ is $0.893069$. Hence, it does not belong to $\mathfrak{AC}$. But, the conditional entropy of $\rho_{2}^{final}$ is $0.331117$, which is positive.
The application of the global unitary operator $U$ (eq.\ref{unitary}) results in a modified state $\rho_{new}^{2}=U\rho_{2}^{final}U^{\dagger}$, given by,
\begin{align}
\rho_{new}^{2}=0.314296|00\rangle\langle 00|+0.392997|00\rangle\langle 11|+0.110924|01\rangle\langle 01|+0.0616247|10\rangle\langle 10|+0.392997|11\rangle\langle 00|+0.513156|11\rangle\langle 11|,
\end{align}

 and the conditional von-Neumann entropy of the state is -0.0620408. Hence, we observe that with the help of a swapping network protocol and a suitable global unitary operator, we got a state whose conditional von-Neumann entropy is negative.

 Again let us consider $\rho_2^{int}$ and another amplitude damped state obtained by passing two qubits of $|\phi^{+}\rangle$ through two amplitude damping channels specified by $p_1$$=$$0.8,$ and $p_2$$\in$$[0.05,0.6].$ Let this initial state be denoted by $\rho_1^{int}.$ Both the initial states $\rho_1^{int}$ and $\rho_2^{int}$ are members of $\mathfrak{AC}$. Now, let these states be used in our swapping protocol. Interestingly, corresponding to any of the four possible outputs of the intermediate party, the conditional state shared between the extreme parties does not belong to $\mathfrak{AC}$ for $p_2$$\in$$[0.05,0.1)$ (see Fig.\ref{fignew}).
\begin{figure}[!ht]
	\centering
	\begin{subfigure}[b]{0.4\textwidth}
		\centering
		\includegraphics[width=\textwidth]{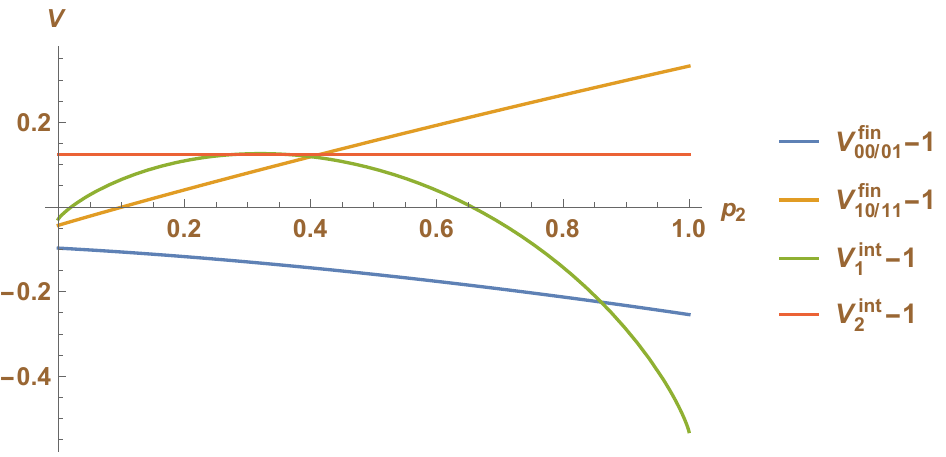}
		\caption{}
	\end{subfigure}
	\hfill
	\begin{subfigure}[b]{0.4\textwidth}
		\centering
		\includegraphics[width=\textwidth]{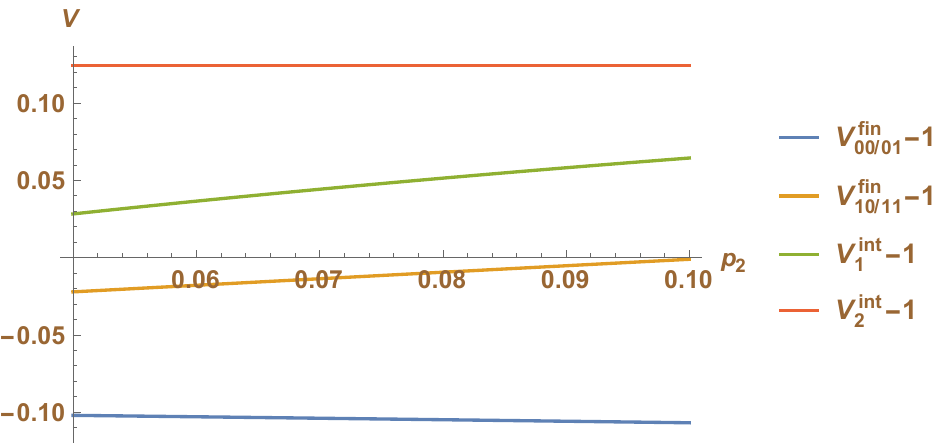}
		\caption{}
	\end{subfigure}
	 \caption{\emph {Von Neumann entropy (V) is plotted against noise parameter  $p_2$ in both the subfigures. For $p_1$$=$$0.8,$ $p_3$$=$$0.2,$ $p_4$$=$$0.714286,$ von Neumann entropy of both the initial states $\rho_1^{int},\rho_2^{int}$ and that of the  final conditional states $\rho^{(ij)}_{fin}$ obtained when the intermediate party gets output $(i,j),$ are plotted here. $V_1^{int}, V_2^{int}$ denote the von Neumann entropy of $\rho_1^{int},\rho_2^{int}$ respectively. $V_{00/01}^{fin}, V_{10/11}^{fin}$ denote the von Neumann entropy of $\rho^{(00)/(01)}_{fin}$ and $\rho^{(10)/(11)}_{fin}$ respectively. Subfig.a shows that for the entire range of the noise parameter $p_2,$  two of the conditional states  $\rho^{(00)/(01)}_{fin}$ are not members of $\mathfrak{AC}$  whereas  one of $\rho_1^{int},\rho_2^{int}$ is a member of $\mathfrak{AC}$. More interestingly, subfig.b shows that there exists range of $p_2$ for which none of the four possible conditional states is a member of $\mathfrak{AC}$ even when both the initial states are so. This in turn points out the effectiveness of our protocol to retrieve back non classicality in the form of negative conditional entropy. }}
	\label{fignew}
\end{figure}

\section{Conclusion}\label{con}

\par Every real-world application of a quantum information processing protocol involves environment interactions. In the present work, we made a probe to see the change in two figure of merits, namely FEF and negative conditional entropy upon environmental interaction. The investigation is to see the transition from a non-absolute regime to an absolute regime, where the term absoluteness is used in the context of global unitary action. We observe that under the action of some quantum channels the state moves to a regime where even a global unitary action fails to retrieve the property. A prescription based on entanglement swapping network protocol that stochastically retrieves the properties have been laid here, towards the purpose. An application of such a prescription is detailed. We have further characterized a class of states whose conditional Rényi entropy remains non-negative under global unitary operations. This further led us to characterize the marginals of a three qudit system with respect to the above-mentioned absolute property. 
\par Our work also provide some usefulness in future direction of work. In \cite{Preskill,Zhou1}, the authors have introduced a method for detecting bipartite entanglement in a many-body mixed state based on estimating moments of the partially transposed density matrix. In this manuscript our main aim is to discuss about two notions of FEF and conditional entropy of quantum states in purview of applying global unitary operations. To be more specific, we have intended to analyze situations where such notions of quantumness enter in the absolute regime due to environmental interactions and when such forms of quantumness can again be retrieved back(in some cases) via entanglement swapping protocol. In this context we are not focussing on adopting strategies based on performing local random measurements and classical post-processings for estimating these properties of quantum states and consider the same as a potential direction of future research. In \cite{Zhou2,Zhou3}  the authors have proposed a systematic method using very few local measurements to detect multipartite entanglement structures based on the graph states. Developing similar strategy to extend our results in section.VI from tripartite to multipartite scenario will be significant. Such an extension can be considered as an interesting direction of future research.

\par FEF and conditional entropies provide for significant benchmarks in quantum information processing protocols. Investigations concerning them are relevant and important, specifically concerning cases when there is a transformation of the density matrix on environmental influence. We believe our work can lend itself to such studies in future.\\

\textbf{Acknowledgement} Tapaswini Patro would like to acknowledge the support from DST-Inspire (INDIA) fellowship No. DST/INSPIRE Fellowship/2019/IF190357. Nirman Ganguly acknowledges support from the project grant received under the SERB(INDIA)-MATRICS scheme vide file number MTR/2022/000101.\\

\section*{Declaration}
All the authors contributed equally to the manuscript. \\
\textbf{Data availability statement:} Data sharing was not applicable to this article as no data sets were generated or analyzed during the current study. \\
\textbf{Competing interests:} The authors have no competing interests to declare. All co-authors have seen and agree with the contents of the manuscript, and there is no financial interest to report.

\bibliography{basename of .bib file}

\end{document}